\documentclass[11pt, reqno]{article}
\pdfoutput=1

\usepackage{graphicx}
\usepackage{jheppub}
\usepackage{amssymb}
\usepackage{amsmath}
\usepackage[usenames,dvipsnames]{xcolor}
\usepackage{epsfig}
\usepackage{dcolumn}
\usepackage{tikz}
\usetikzlibrary{shapes.geometric, arrows}
\usepackage{upgreek}
\usepackage{setspace}
\usepackage{enumitem}
\usepackage{array,multirow,bigdelim,arydshln}
\usepackage{appendix}
\usepackage{xparse}
\usepackage[utf8]{inputenc}
\usepackage{hyperref}
\hypersetup{
	colorlinks,
	urlcolor=Maroon,
	linkcolor=Maroon,
	citecolor=Maroon
}

\usepackage{amsthm}
\usepackage{amsmath}
\newtheorem{thm}{Theorem}[section]
\newtheorem{prop}[thm]{Proposition}

\newtheorem{defn}[thm]{Definition}
\theoremstyle{definition}

\newtheorem{conjecture}[thm]{Conjecture}

\usepackage{mathtools}
\usepackage{scalerel}

\usepackage{float}
\restylefloat{table}

\NewDocumentCommand{\binomial}{omm}
{%
	\genfrac(){0pt}{}{#2}{#3}%
	\IfValueT{#1}{_{\!#1}}%
}
\NewDocumentCommand{\eulerian}{omm}
{%
	\genfrac<>{0pt}{}{#2}{#3}%
	\IfValueT{#1}{_{\!#1}}%
}

\def \s {\sigma}

\usepackage{underscore}

\usepackage{latexsym}
\usepackage{tikz}

\title{Biadjoint Scalars and Associahedra from Residues of Generalized Amplitudes}
\author[a]{Freddy Cachazo}
\author[b]{and Nick Early}
\affiliation[a]{Perimeter Institute for Theoretical Physics, Waterloo, ON N2L 2Y5, Canada.}
\affiliation[b]{Max Planck Institute for Theoretical Physics, F\"{o}hringer Ring 6, Munich, Germany.}
\emailAdd{fcachazo@pitp.ca}
\emailAdd{earlnick@mpp.mpg.de}

\abstract{
In the Grassmannian formulation of the S-matrix for planar $\mathcal{N}=4$ Super Yang-Mills, $N^{k-2}MHV$ scattering amplitudes for $k$ negative and $n-k$ positive helicity gluons can be expressed, by an application of the global residue theorem, as a signed sum over a collection of $(k-2)(n-k-2)$-dimensional residues.  These residues are supported on certain positroid subvarieties of the Grassmannian $G(k,n)$.  In this paper, we replace the Grassmannian $G(3,n)$ with its torus quotient, the moduli space of $n$ points in the projective plane in general position, and planar $\mathcal{N}=4$ SYM with generalized biadjoint scalar amplitudes $m^{(3)}_n$ as introduced by Cachazo-Early-Guevara-Mizera (CEGM) \cite{Cachazo:2019ngv}.  Whereas in the Grassmannian formulation residues of the Parke-Taylor form correspond to individual BCFW, or on-shell diagrams, we show that each such $(n-5)$-dimensional residue of $m^{(3)}_n$ is an entire biadjoint scalar partial amplitude $m^{(2)}_n$, that is a sum over all tree-level Feynman diagrams for a fixed planar order.  We make a proposal which would do the same for $k\ge 4$.   

}

\begin{document}
	\maketitle
	\addtocontents{toc}{\protect\setcounter{tocdepth}{1}}
	\def \tr {\nonumber\\}
	\def \nn {\nonumber}
	\def \la {|}
	\def \ra {|}
	\def \lan {\langle}
	\def \ran {\rangle}
	\def \dd {\Theta}
	\def\hset{\texttt{h}}
	\def\gset{\texttt{g}}
	\def\sset{\texttt{s}}
	\def \be {\begin{equation}}
		\def \ee {\end{equation}}
	\def \ba {\begin{eqnarray}}
		\def \ea {\end{eqnarray}}
	\def \bg {\begin{gather}}
		\def \eeg {\end{gather}}
	\def \k {\kappa}
	\def \h {\hbar}
	\def \r {\rho}
	\def \l {\lambda}
	\def \be {\begin{equation}}
		\def \en {\end{equation}}
	\def \bes {\begin{eqnarray}}
		\def \ens {\end{eqnarray}}
	\def \red {\color{Maroon}}
	\def \pt {{\rm PT}}
	\def \s {\mathfrak{s}}
	\def \t {\mathfrak{t}}
	\def \v {\mathfrak{v}}
	\def \C {\textsf{C}}
	\def \tp {||}
	\def \p {x}
	\def \x {z}
	\def \V {\textsf{V}}
	\def \ls {{\rm LS}}
	\def \ma {\Upsilon}
	\def \SL {{\rm SL}}
	\def \GL {{\rm GL}}
	\def \w {\omega}
	\def \e {\epsilon}
	
	\numberwithin{equation}{section}

	
	\section{Introduction}
	
	Biadjoint scalar amplitudes, $m^{(2)}_n$, are the simplest amplitudes that admit a Cachazo-He-Yuan (CHY) formulation \cite{Cachazo:2013gna,Cachazo:2013hca,Cachazo:2013iea,Dolan:2013isa}.  These doubly flavor-ordered tree-level amplitudes, $m_n^{(2)}$, are given as integrals over the configuration space of $n$ points on $\mathbb{CP}^1$, usually denoted by $X(2,n)$. More generally, we denote by $X(k,n)$ the moduli space of $n$ generic points on $\mathbb{CP}^{k-1}$.

	The generalization of the CHY formula for biadjoint scalar amplitudes to integrals over $X(k,n)$ leads to CEGM amplitudes \cite{Cachazo:2019ngv}, $m_n^{(k)}$. Since their discovery, a number of properties that mirror those of $\mathcal{N}=4$ super Yang-Mills (SYM) amplitudes have been found. For example, it is now known that generalized amplitudes admit ``soft'' and ``hard limits" \cite{GarciaSepulveda:2019jxn}, and in fact, the limits are mapped to each other under the duality $(k,n)\sim (n-k,n)$. Under the limits, $m_n^{(k)}$ gives rise to either $m_{n-1}^{(k)}$ or $m_{n-1}^{(k-1)}$ respectively. This behaviour is reminiscent of soft limits of $(k,n)$ ${\cal N}=4$ SYM amplitudes which give rise to $(k,n-1)$ amplitudes or $(k-1,n-1)$ amplitudes depending on the helicity of the soft gluon. Moreover, a parity transformation takes $(k,n)$ SYM amplitudes into $(n-k,n)$ SYM amplitudes and thus exchanges the two soft limits.

	Our main result is the proof that $m_n^{(2)}$ can be obtained via a multidimensional residue of $m_n^{(3)}$ in ${\rm dim}\, X(3,n)-{\rm dim}\, X(2,n)= n-5$ kinematic invariants. More precisely, in Theorem \ref{residual3} we prove that
	\be 
	{\rm Res}[m^{(3)}_n]_{(\t_{45\ldots n}=0,\t_{56\ldots n}=0,\cdots ,\s_{n-2,n-1,n}=0)} = m^{(2)}_n
	\ee 
	for some appropriate identification of kinematic invariants presented in \eqref{fullmap}.  
	
	We also conjecture a way to extend our procedure to all $k$ that requires a $(k-2)(n-k-2)$-dimensional residue.  The precise coincidence with the Grassmannian formulation of ${\cal N}=4$  SYM scattering amplitudes is not surprising since 
	\be 
	{\rm dim}\, G(k,n)-{\rm dim}\, G(2,n)  = {\rm dim}\, X(k,n)-{\rm dim}\, X(2,n) = (k-2)(n-k-2). 
	\ee 
	This follows from ${\rm dim}\, G(k,n) = {\rm dim}\, X(k,n)+(n-1)$.
	
	Returning to $k=3$, the ways of choosing the $(n-5)$-dimensional residue have a interesting combinatorial structure. We find exactly $n\cdot C_{n-5}$ ways, with $C_p$ the $p^{\rm th}$ Catalan number, one for each choice of $n-5$ consecutive labels and each triangulation of a $(n-3)$-gon. The bijection between such triangulations and Feynman diagrams is a hint that a $m^{(2)}_{n-3}$ amplitude could be constructed. Indeed, by specializing the kinematics to what we call a ``parallel hard'' limit, we find that $m^{(3)}_n\sim m^{(2)}_{n-3}m^{(2)}_n$. 
	
	The reader familiar with Grassmannian formulations of SYM amplitudes would also recall that there is an alternative version based on the Veronese embedding of $\mathbb{CP}^1$ into $\mathbb{CP}^{k-1}$ which is equivalent to the Witten-RSV formulation \cite{Witten:2003nn,Roiban:2004yf}. In this work we also start the study of modified CEGM amplitudes which have residues on poles that enforce the Veronese condition on the points in $\mathbb{CP}^2$. Moreover, we show for $(3,6)$ that the residue coincides with $m^{(2)}_6$.

	The \textit{existence} of residues of $m^{(3)}_n$ which can be identified with $m^{(2)}_n$, would perhaps not be surprising to experts; see for example \cite{Arkani-Hamed:2019mrd, Arkani-Hamed:2019plo} as well as \cite{Abhishek:2020sdr}, where the forward limit of $m^{(2)}_6$ was identified with a facet of the D4 cluster polytope.

	With regard to the all $k\ge 3$ generalization we formulate an explicit identification of $m^{(2)}_n$ with a particular $(k-2)(n-k-2)$-dimensional residue of $m^{(k)}_n$, respectively $m^{(k,NC)}_n$, which directly generalizes our $k=3$ result; therefore the $k=3$ case has been proved in this paper, see Theorem \ref{residual3}.  The generalization is stated for all $3\le k\le n-3$.

	Additionally, our identification of $(n-5)$-dimensional residues using the CEGM formula predicts a novel formula for the associahedron, as a Minkowski sum involving positroid polytopes in the second hypersimplices $\Delta_{2,n}$; the end result is a generalized permutohedron, as in the realization due to Loday \cite{loday2004realization}, but the presence of the rank two positroids is new.  It seems reasonable to expect that a story will emerge for $k\ge 4$ as well, involving positroid polytopes of higher ranks.
	
	This note is organized as follows: In section \ref{sec: CHY CEGM review} we review the CHY and CEGM formulas for $m^{(2)}_n$ and $m^{(k)}_n$ respectively. In section \ref{particularR} we state and prove the main result of this work by computing a $(n-5)$-dimensional residue of $m^{(3)}_n$ using the CEGM formula and obtaining $m^{(2)}_n$. In section \ref{sec: paths emergent amp} we describe the different choices of residues and the emergence of $m_{n-3}^{(2)}$. This leads to a proposal for all $k$ which is explicitly defined in section \ref{sec: generalization higher k proposal}.  Also in section \ref{sec: generalization higher k proposal}, we explore connections to associahedra and the noncrossing amplitude $m^{(3,NC)}_n$.  In section \ref{veronese} we present the Veronese generalization of biadjoint amplitudes and compute the residue that leads to $m^{(2)}_6$. We end with discussion about relations to other constructions and future directions in section \ref{sec: discussions}.  
	
	\section{Review of CHY and CEGM Formulations}\label{sec: CHY CEGM review}
	
	The biadjoint scalar theory has flavour group $U(N)\times U({\tilde N})$, a massless scalar field in the biadjoint representation $\phi^{a\tilde a}$ and cubic interactions governed by the structure constants of the flavour group. Tree-level scattering amplitudes can be flavour decomposed and computed via the standard Feynman diagrammatic technique or the CHY formulation. We refer the reader to \cite{Cachazo:2013gna,Cachazo:2013hca,Cachazo:2013iea,Dolan:2013isa} for a more in-depth discussion. Here we simply review the formula for the amplitude. The CHY formula is an integral over the configuration space of $n$ generic points on $\mathbb{CP}^1$. A convenient parameterization of $X(2,n)$ is  
	\be\label{first2n} 
	M := \left[ \begin{array}{ccccccccc}
		1 & 0 & 1 & y_4 & y_5   & \cdots & y_{n-1} & y_n  \\
		0 & 1 & 1 & 1   & 1     & \cdots & 1 & 1
	\end{array}  \right] .
	\ee 
	The CHY formula is given by
	\be\label{intCHY} 
	m^{(2)}(\alpha, \beta ) =\int \prod_{a=4}^{n}dy_a\,\delta\left( \frac{\partial {\cal S}^{(2)}_n}{\partial y_a}\right)\, {\rm PT}(\alpha ){\rm PT}(\beta )
	\ee 
	where 
	\be\label{chyFirst} 
	{\cal S}^{(2)}_n := \sum_{a<b} s_{ab}\log \Delta_{ab}, \quad \Delta_{ab} := {\rm det}\left[ \begin{array}{cc}
		M_{1a} & M_{1b} \\
		M_{2a} & M_{2b}
	\end{array}  \right]
	\ee 
	and the so-called Parke-Taylor functions are defined for an ordering $\alpha = (\alpha_1 ,\alpha_2, \ldots ,\alpha_n )$ of $[n]$,
	\be 
	{\rm PT}(\alpha ) := \frac{1}{\Delta_{\alpha_1,\alpha_2}\Delta_{\alpha_2,\alpha_3}\cdots \Delta_{\alpha_n,\alpha_1}}.
	\ee 
	The integrals in \eqref{intCHY} are localized to the $(n-3)!$ solutions to the scattering equations and therefore,
	\be\label{intCHY2} 
	m^{(2)}(\alpha, \beta ) =\sum_{\rm solutions} \left({\rm det}\left[ \frac{\partial^2 {\cal S}^{(2)}_n}{\partial y_a\partial y_b}\right]\right)^{-1} \, {\rm PT}(\alpha ){\rm PT}(\beta ) .
	\ee 
	In order to make the permutation invariance of the measure manifest, CHY introduced \cite{Cachazo:2013gna,Cachazo:2013hca,Cachazo:2013iea} a reduced jacobian, ${\rm det}'\Phi$, which in our parameterization \eqref{first2n} becomes
	\be 
	{\rm det}'\Phi ={\rm det}\left[ \frac{\partial^2 {\cal S}^{(2)}_n}{\partial y_a\partial y_b}\right].
	\ee 

	In 2019, CEGM generalized the CHY formula to construct ``generalized biadjoint amplitudes'' as integrals over $X(k,n)$ \cite{Cachazo:2019ngv}. Of course, since $X(n-2,n)\sim X(2,n)$, the CEGM formulas for $k=n-2$ give a new but equivalent representation of the standard amplitudes, i.e. $m^{(2)}_n = m^{(n-2)}_n$ up to a proper identification of kinematic invariants. 
	
	We refer the reader to \cite{Cachazo:2019ngv} for more details. Here we simply note that using the parameterization of $X(3,n)$ given by %
	\be\label{cegm3para} 
	M := \left[ \begin{array}{cccccccc}
		1 & 0 & 0 & 1 & y_{15}  & y_{16} & \cdots & y_{1n}  \\
		0 & 1 & 0 & 1 &  y_{25}  & y_{26} & \cdots & y_{2n} \\
		0 & 0 & 1 & 1 &    1  & 1   & \cdots & 1 
	\end{array}  \right] ,
	\ee 
	one finds 
	\be\label{cegm3n}
	m^{(3)}_n(\alpha,\beta):= \sum_{\rm solutions} \left({\rm det}\left[ \frac{\partial^2 {\cal S}^{(3)}_n}{\partial y_{ia}\partial y_{jb}}\right]\right)^{-1} \, {\rm PT}^{(3)}(\alpha ){\rm PT}^{(3)}(\beta ),
	\ee
	where 
	\be\label{cegmPot} 
	{\cal S}^{(3)}_n := \sum_{a<b<c} \s_{abc}\log \Delta_{abc}, \quad \Delta_{abc} := {\rm det}\left[ \begin{array}{ccc}
		M_{1a} & M_{1b} & M_{1c}\\
		M_{2a} & M_{2b} & M_{2c}\\
		M_{3a} & M_{3b} & M_{3c}
	\end{array}  \right] ,
	\ee 
	and
	\be 
	{\rm PT}(\alpha ) := \frac{1}{\Delta_{\alpha_1,\alpha_2,\alpha_3}\Delta_{\alpha_2,\alpha_3,\alpha_4}\cdots \Delta_{\alpha_{n-1},\alpha_n,\alpha_1}\Delta_{\alpha_n,\alpha_1,\alpha_2}}.
	\ee 
	The generalized kinematic Mandelstam invariants in $\mathcal{S}^{(3)}_n$ are defined by the following conditions,
	\be
	\s_{abc} = \s_{bac} = \s_{acb},\quad \s_{aab}=0,\quad \sum_{b,c=1}^n\s_{abc} =0 \quad \forall a.  
	\ee
	It is also useful to define
	\be
	\t_J := \sum_{\{a,b,c\} \subset J}\s_{abc}.
	\ee
	For example, $\t_{1234}=\s_{123}+\s_{124}+\s_{134}+\s_{234}$.
	
	More generally, the formula for any $k$ is easily obtained in a completely analogous parameterization using $y_{ia}$ with $i\in \{1,2,\ldots ,k-1\}$ and $a\in \{ k+2,k+3,\ldots ,n\}$.

	\section{From $k=3$ CEGM Integrals to CHY Integrals}\label{particularR}
	
	In this section we present the explicit computation of one of the multi-dimensional residues which connects $m^{(3)}_n$ to $m^{(2)}_n$. In this work we only deal with partial amplitudes where both orderings are the canonical order. We therefore abuse notation and write
	\be  
	m^{(k)}_n := m^{(k)}_n(\mathbb{I},\mathbb{I}). 
	\ee 

	The main result of this work is the following theorem. 
	
	\begin{thm}\label{residual3}
		The $(n-5)$-dimensional residue of $m_n^{(3)}$ defined by 
		\be\label{poles} 
		(\t_{45\ldots n},\ \t_{56\ldots n},\ \t_{67\ldots n},\ldots, \s_{n-2,n-1,n})=(0,0,0,\ldots , 0)
		\ee
		is equal to $m_n^{(2)}$ under some identification of kinematic invariants.
	\end{thm}
	
	The proof is simple once a convenient parameterization of $X(3,n)$ is found. We start with one which makes the final codimension $n-5$ configuration of points of interest in $\mathbb{CP}^2$ manifest,
	\be\label{embed2n} 
	\left[ \begin{array}{ccccccccc}
		1 & 0 & 0 & 1 & x_2 & x_2-x_3  & x_2-x_4 & \cdots & x_2-x_{n-3}  \\
		0 & 1 & 0 & 1 & x_1 & x_1-x_3  & x_1-x_4 & \cdots & x_1-x_{n-3} \\
		0 & 0 & 1 & 1 &  1  &  1-x_3   & 1-x_4   & \cdots & 1 - x_{n-3} 
	\end{array}  \right] .
	\ee 
	Note that the $r^{\rm th}$-column, with $r>5$, is given by a linear combination of the $4^{\rm th}$ and $5^{\rm th}$ columns. This is precisely the configuration of points in which points $4,5,\ldots ,n$ lie on a line in $\mathbb{CP}^2$.
	
	The next step is to introduce $n-5$ deformation parameters $\{ z_1,z_2,\ldots ,z_{n-5}\}$ to obtain a complete parameterization of $X(3,n)$, 
	\be\label{repreM} 
	\left[ \begin{array}{ccccccccc}
		1 & 0 & 0 & 1+z_1 & x_2+z_1z_2 & x_2-x_3+z_1z_2z_3  & x_2-x_4+z_1z_2z_3z_4 & \cdots & x_2-x_{n-3}  \\
		0 & 1 & 0 & 1 & x_1 & x_1-x_3  & x_1-x_4 & \cdots & x_1-x_{n-3} \\
		0 & 0 & 1 & 1 &  1  &  1-x_3   & 1-x_4   & \cdots & 1 - x_{n-3} 
	\end{array}  \right] .
	\ee 
	The CEGM construction proceeds by connecting the space of kinematic invariants with $X(3,n)$ by computing the critical points of the function,
	\be 
	{\cal S}^{(3)}_n := \sum_{a<b<c}\s_{abc} \log \Delta_{abc} ,
	\ee 
	where $\Delta_{abc}$ are the minors of the matrix representative of a point in $X(3,n)$.
	
	Computing ${\cal S}^{(3)}_n$ and expanding reveals the following structure
	\be\label{struc} 
	{\cal S}^{(3)}_n = \ldots + \t_{45\ldots n}\log z_1 + \t_{56\ldots n}\log z_2 + \t_{67\ldots n}\log z_3 + \ldots +\s_{n-2,n-1,n} \log z_{n-5}. 
	\ee 
	Here the first ellipses stand for terms that remain finite as $z_i=0$. This form shows that as a given kinematic invariant is taken to zero the corresponding deformation parameter must as well vanish. In order to make the connection manifest, let us introduce the notation $\{ J_1,\ldots ,J_{n-5}\}$ for the sets of labels relevant in \eqref{struc}, i.e. 
	\be  
	\sum_{i=1}^{n-5} \t_{J_i}\log z_i :=\t_{45\ldots n}\log z_1 + \t_{56\ldots n}\log z_2 + \t_{67\ldots n}\log z_3 + \ldots +\s_{n-2,n-1,n} \log z_{n-5}. 
	\ee 

	Next, consider the scattering equations, i.e., the equations for computing the critical points of ${\cal S}^{(3)}_n$. Recalling that we are interested in the behavior near the subspace defined by \eqref{poles}, the equations can be separated into two groups, 
	\be\label{GxGz} 
	G_x=\left\{\frac{\partial {\cal S}^{(3)}_n}{\partial x_i} =0,\quad i\in \{1,2,\ldots ,n-3\}\right\},\quad G_z=\left\{\frac{\partial {\cal S}^{(3)}_n}{\partial z_i} =0,\quad i\in \{1,2,\ldots ,n-5\}\right\}.
	\ee
	It is convenient to keep all kinematics invariants that do not participate in \eqref{poles} fixed and so treat \eqref{poles} as the equations that define a point in $\mathbb{C}^{n-5}$. In the remainder of this section we adopt this viewpoint.
	
	At the point \eqref{poles} the first set, $G_x$, becomes a system of equations for only $x$-variables, i.e., all $z$-variables drop out. These will be shown to be nothing but the scattering equations of a $k=2$ CHY function. Using \eqref{struc}, the equations in $G_z$ have terms of the form $\t_{J_i}/z_i$ indicating that as $\t_{J_i}\to 0$ the ratio should be kept finite and therefore we write $z_i = \t_{J_i}q_i$. In the limit as the point \eqref{poles} is approached, the $G_z$ equations simplify and become
	\be\label{singular}  
	\frac{1}{q_i} + F_i(x) = 0, \quad i\in \{1,2, \ldots ,n-5\} ,
	\ee 
	for some functions $F_i(x)$. Therefore for each solution to the $G_x$ equations one finds a unique solution to the $G_z$ equations. 
	
	The final step is to compute the Jacobian matrix, usually denoted by, ${\rm det}'\Phi^{(3)}_n$. This is closely related to the determinant of the Hessian matrix of ${\cal S}^{(3)}_n$. In order to make the connection precise, it is most convenient to use $GL(3)$ and torus action transformations to bring the matrix representative \eqref{repreM} into a canonical frame,
	\be\label{canon} 
	\left[ \begin{array}{ccccccccc}
		1 & 0 & 0 & 1 & \frac{x_2+z_1z_2}{1+z_1} & \frac{x_2-x_3+z_1z_2z_3}{(1+z_1)(1-x_3)}  & \frac{x_2-x_4+z_1z_2z_3z_4}{(1+z_1)(1-x_4)} & \cdots & \frac{x_2-x_{n-3}}{(1+z_1)(1 - x_{n-3})}  \\
		0 & 1 & 0 & 1 & x_1 & \frac{x_1-x_3}{1-x_3}  & \frac{x_1-x_4}{1-x_4} & \cdots & \frac{x_1-x_{n-3}}{1 - x_{n-3}} \\
		0 & 0 & 1 & 1 &  1  &  1  & 1   & \cdots & 1 
	\end{array}  \right] .
	\ee 
	In this frame ${\rm det}'\Phi^{(3)}_n$ is the determinant of the Hessian matrix of ${\cal S}^{(3)}_n$. Ordering the variables as $\{x_1,x_2,\ldots, x_{n-3},z_1,z_2,\ldots, z_{n-5}\}$, the matrix of second derivatives has the block form
	\be\label{hessian}
	\left[ \begin{array}{ccc|ccc}  & & & & &  \\ & \frac{\partial^2 {\cal S}^{(3)}_n}{\partial x_i\partial x_j} & & & \frac{\partial^2 {\cal S}^{(3)}_n}{\partial x_i\partial z_j} & \\ & & & & &  \\ \hline & & & & &  \\ &
		\frac{\partial^2 {\cal S}^{(3)}_n}{\partial z_i\partial x_j} & & & \frac{\partial^2 {\cal S}^{(3)}_n}{\partial z_i\partial z_j} & \\ 
		& & & & &   \end{array} \right] .
	\ee 
	Using the form of the $G_z$ equations \eqref{singular} it is easy to see that the off-diagonal blocks and the top left block are all finite as the point \eqref{poles} is approached while the bottom right block becomes 
	\be 
	\left[ \frac{\partial^2 {\cal S}^{(3)}_n}{\partial z_i\partial z_j} \right] = 
	\left[ \begin{array}{cccc} 
		\frac{1}{\t_{J_1}q_1^2} & * & \cdots & * \\ 
		* & \frac{1}{\t_{J_2}q_2^2} & \cdots & * \\
		* & * & \cdots &\frac{1}{\t_{J_{n-5}}q_{n-5}^2} \\ 
	\end{array}  \right] .
	\ee 
	Here $*$ denotes terms that remain finite.
	
	Since we are only interested in the leading order divergence near the point \eqref{poles}, the determinant of \eqref{hessian} is given by
	\be\label{weri} 
	\left(\prod_{i=1}^{n-5}\frac{1}{\t_{J_i}q_i^2}\right){\rm det}\left[ \frac{\partial^2 {\cal S}^{(3)}_n}{\partial x_i\partial x_j} \right] .
	\ee 
	The goal is to show that the remaining determinant is exactly that of the CHY formula for $m^{(2)}_n$.
	
	A preliminary step is to identify the parameterization of $X(2,n)$ that corresponds to \eqref{embed2n}. This can be done by analysing the structure of the $k=3$ Parke-Taylor factor and the way it reduces to a $k=2$ object.
	
	The $k=3$ Parke-Taylor function is defined as
	\be  
	{\rm PT}^{(3)}:= \frac{1}{\Delta_{123}\Delta_{234}\cdots \Delta_{n12}}.
	\ee 
	In order to connect to the CEGM formula for $m^{(3)}_n$ one has to compute the jacobian of the change of variables from the entries of the matrix in \eqref{canon} to 
	\be\label{cegm} 
	\left[ \begin{array}{ccccccccc}
		1 & 0 & 0 & 1 & y_{15}  & y_{16} & \cdots & y_{1n}  \\
		0 & 1 & 0 & 1 &  y_{25}  & y_{26} & \cdots & y_{2n} \\
		0 & 0 & 1 & 1 &   1  & 1   & \cdots & 1 
	\end{array}  \right] .
	\ee 
	Denoting this jacobian as ${\bf J}_{(x,z):y}$ and using \eqref{weri} one finds 
	\be 
	{\rm det}'\Phi^{(3)}_n = 
	\left(\prod_{i=1}^{n-5}\frac{1}{\t_{J_i}q_i^2}\right){\rm det}\left[ \frac{\partial^2 {\cal S}^{(3)}_n}{\partial x_i\partial x_j} \right] \left({\bf J}_{(x,z):y}\right)^{-2}.
	\ee
	
	The CEGM formula is given by
	\be 
	m^{(3)}_n = \sum_{{\rm solutions}~{\rm of}~G_x,G_z} \frac{1}{{\rm det}'\Phi^{(3)}_n}\left({\rm PT}^{(3)}\right)^2.  
	\ee
	Substituting the formula for ${\rm det}'\Phi^{(3)}_n$ one finds the leading order
	\be\label{almost} 
	\left. m^{(3)}_n\right|_{\rm Leading\, order} = \sum_{{\rm solutions}~{\rm of}~G_x,G_z} \left(\prod_{i=1}^{n-5}\t_{J_i}q_i^2\right)\left({\rm det}\left[ \frac{\partial^2 {\cal S}^{(3)}_n}{\partial x_i\partial x_j} \right]\right)^{-1}\left( {\bf J}_{(x,z):y} {\rm PT}^{(3)}\right)^2.
	\ee 
	A straightforward computation reveals that 
	\be\label{ppt} 
	\left( {\bf J}_{(x,z):y} {\rm PT}^{(3)}\right)^2 =\left(\prod_{i=1}^{n-5}\frac{1}{\t_{J_i}q_i}\right)^2\left(\frac{1}{(x_1-1)(x_1-x_2)x_3\prod_{j=3}^{n-4}(x_i-x_{i+1})(x_{n-3}-1)}\right)^2 
	\ee 
	This means that if we could identify the Hessian in \eqref{almost} with ${\rm det}'\Phi^{(2)}_n$ and the second factor in \eqref{ppt} with ${\rm PT}^{(2)}$, then we would have succeeded in relating $m^{(3)}_n$ to $m^{(2)}_n$. Once again, this is easily done once an appropriate parameterization of $X(2,n)$ is found. 
	
	Consider the following  
	\be\label{para2n} 
	\left[ \begin{array}{ccccccccc}
		1 & 0 & 1 & x_2 & x_2-x_3  & x_2-x_4 & \cdots & x_2-x_{n-3} & x_2-1  \\
		0 & 1 & 1 & x_1 & x_1-x_3  & x_1-x_4 & \cdots & x_1-x_{n-3} & x_1-1
	\end{array}  \right] .
	\ee 
	In order to connect to the CHY formula for $m^{(2)}_n$ the simplest approach is to bring \eqref{para2n} to a canonical frame, and then compute the jacobian of the change of variables to the entries of the matrix ${\bf J}^{(2)}_{x:y}$ and the Parke-Taylor function
	\be  
	{\rm PT}^{(2)}:= \frac{1}{\Delta_{12}\Delta_{23}\cdots \Delta_{n1}}.
	\ee 
	The combination 
	$${\bf J}^{(2)}_{x:y}\, {\rm PT}^{(2)} =\frac{1}{(x_1-1)(x_1-x_2)x_3\prod_{j=3}^{n-4}(x_i-x_{i+1})(x_{n-3}-1)} $$
	which is precisely what is needed. 
	
	The last step is to identify how the function ${\cal S}^{(3)}_n$ is related to 
	\be\label{chyPot2} 
	{\cal S}^{(2)}_n = \sum_{a<b} s_{ab}\log \Delta_{ab}
	\ee 
	for some choice of $s_{ab}$. 
	
	We show the explicit map in the next subsection so let us close this discussion by summarizing. We have shown that near the point \eqref{poles}, the leading order behaviour of $m^{(3)}_n$ is 
	\be\label{almostT} 
	\left. m^{(3)}_n\right|_{\rm Leading\, order}  = \left(\prod_{i=1}^{n-5}\frac{1}{\t_{J_i}}\right)\sum_{{\rm solutions}~{\rm of}~G_x} \left({\rm det}\left[ \frac{\partial^2 {\cal S}^{(3)}_n}{\partial x_i\partial x_j} \right]\right)^{-1}\left( {\bf J}^{(2)}_{x:y} {\rm PT}^{(2)}\right)^2.
	\ee 

	\subsection{Identification of Kinematic Invariants}
	
	We now have to show that in the limit $\t_{J_i}\to 0$ and with $z_i =\t_{J_i}q_i$, the CEGM function 
	\be 
	{\cal S}^{(3)}_n := \sum_{a<b<c}\s_{abc}\rm \log \Delta_{abc}
	\ee 
	maps to its $k=2$ version \eqref{chyPot2}. These functions are invariant under their corresponding ${\rm GL}(k,n)$ and torus actions so we can choose the presentation of their matrix representatives. For $X(3,n)$ we work with \eqref{repreM}.
	
	The limit $\t_{J_i}\to 0$ is taking by expanding around $z_i=\t_{J_i} q_i$. The limit is finite since the only potentially singular pieces are those extracted in \eqref{struc}. Explicitly computing the structure of the minors one discovers that 
	\be\label{strucLimit} 
	{\rm lim}_{\t_{J_i}\to 0}\, {\cal S}^{(3)}_n \! =\! \sum_{a=4}^n \s_{12a}\log (x_{a-3}-1) + \sum_{a=4}^n \s_{13a}\log (x_{a-3}-x_1)+\sum_{a=4}^n \s_{23a}\log (x_{a-3}-x_2)+\ldots.
	\ee 
	
	Every single minor becomes a rational function of $x_a$, $x_a-1$ or $x_a-x_b$. This is exactly the structure of ${\cal S}^{(2)}_n$ using the parameterization \eqref{para2n}. 
	
	Imposing that 
	\be\label{potCon} 
	{\rm lim}_{\t_{J_i}\to 0}\, {\cal S}^{(3)}_n ={\cal S}^{(2)}_n
	\ee 
	as functions of $x_i$ implies that the coefficients of each logarithm on each side must agree. This is a system of linear equations for the rank-three kinematic invariants $\s_{abc}$. 
	
	The solution is most easily presented by shifting the labels of the $k=2$ invariants, i.e., $s_{a,b}\to s_{a+1,b+1}$. This is clearly irrelevant as any action of the dihedral group on $n$ elements gives the same $m^{(2)}_n$. The solution is then
	\begin{align}\label{fullmap} 
		&    \s_{123} = s_{12}+s_{13}+s_{23},\quad \s_{12a} = s_{1a},\quad \s_{13a} = s_{2a},\quad  \s_{23a} = s_{3a} \quad a\in \{ 4,5,\ldots ,n\} ,
		\\
		& \nonumber \s_{a,a+1,b} = s_{a+1,b} - \sum_{c=1}^{a-1}\s_{a+1,b,c} \quad a\in \{ 3,4,\ldots ,n-4\} \quad b \in \{ a+2,a+3,\ldots ,n \}. 
	\end{align}

	Using this map of kinematic invariants and \eqref{potCon} it is clear that 
	\be 
	{\rm lim}_{\t_{J_i}\to 0}\, \left({\rm det}\left[ \frac{\partial^2 {\cal S}^{(3)}_n}{\partial x_i\partial x_j} \right]\right) = {\rm det}\left[ \frac{\partial^2 {\cal S}^{(2)}_n}{\partial x_i\partial x_j} \right] = \left({\bf J}^{(2)}_{x:y}\right)^2 {\rm det}'\Phi^{(2)}_n.
	\ee 
	Combining this with \eqref{almostT} the proof of theorem \ref{residual3} is complete since
	\be\label{finn}
	\left. m^{(3)}_n\right|_{\rm Leading\, order}  = \left(\prod_{i=1}^{n-5}\frac{1}{\t_{J_i}}\right)\sum_{{\rm solutions}~{\rm of}~G_x} \left({\rm det}'\Phi^{(2)}_n \right)^{-1}\left( {\rm PT}^{(2)}\right)^2 =\left(\prod_{i=1}^{n-5}\frac{1}{\t_{J_i}}\right)\, m^{(2)}_n .
	\ee
	
	\subsection{Examples}\label{sec: examples}
	
	The simplest possible case is for $n=5$. Clearly no residues are needed since $X(3,5)\sim X(2,5)$. Nevertheless let us illustrate the map of kinematic invariants as it is not the standard one.
	
	Choosing five invariants as basis one finds from \eqref{fullmap}
	\be 
	\s_{123} = s_{12}+s_{23}+s_{13},\quad \s_{124}=s_{14},\quad \s_{125}=s_{15}, \quad \s_{134}=s_{24},\quad 
	\s_{135} = s_{25}.
	\ee 
	Let us check that the planar invariants map correctly by using momentum conservation,
	\be
	\s_{123} = s_{45}, \quad \s_{234} = s_{34}, \quad \s_{345} = s_{23}, \quad \s_{451} = s_{12}, \quad \s_{512} = s_{51}.  
	\ee
	This shows that the map \eqref{fullmap} takes the $(3,5)$ kinematic space into a reflection of the $(2,5)$ kinematic space. Of course, amplitudes are invariant under the dihedral group, even though the kinematic twist is not for $n\ge 6$.
	
	The next example is less trivial. Consider $m^{(3)}_6$ and its residues at $s_{456}=0$. This generalized amplitude was first computed in \cite{Cachazo:2019ngv} and can be written as a sum over $48$ generalized Feynman diagrams. Only $14$ of them possess a pole at $s_{456}=0$. This is exactly the number of Feynman diagrams in $m^{(2)}_6$. Recall that each $k=3$ generalized Feynman diagram has four propagators and therefore after computing the residue at $s_{456}=0$ the fourteen surviving diagrams have three propagators. Moreover, under the map \eqref{fullmap}, each term becomes one $k=2$ Feynman diagram term. This behaviour shows that the relations between residues of $m^{(3)}_n$ and $m^{(2)}_n$ is actually a combinatorial statement about the structures of the corresponding generalized associahedra or positive tropical Grassmannian descriptions. 
	
	In \eqref{fullmap} we have only provided the map for the invariants that appear after the residue is computed modulo momentum conservation. Let us present the complete list of 20 invariants 
	\begin{align}\label{map36to26}
		& \s_{123}= s_{12}+s_{13}+s_{23},\, \s_{124}= s_{14},\,\s_{125}= s_{15},\s_{126}= s_{16}, \,\s_{134}= s_{24},\,\s_{135}= s_{25}, \\ \nonumber
		& \s_{136}= s_{26},\,\s_{156}= s_{12}-\s_{145}-\s_{146},\s_{234}=
		s_{34},\, \s_{235}= s_{35},\, \s_{236}= s_{36}, \\ \nonumber & \s_{256}=
		s_{13}-\s_{245}-\s_{246},\, \s_{345}=
		s_{45}-\s_{145}-\s_{245},\, \s_{346}=
		s_{46}-\s_{146}-\s_{246},\\\nonumber & \s_{356}=
		-s_{12}-s_{13}+s_{56}+\s_{145}+\s_{146}+\s_{245}+\s_{246},\,\s_{456}=
		0.
	\end{align}
	%
	Note that $\s_{145},\s_{146},\s_{245},\s_{246} $ are arbitrary and can be set to zero.
	
	Consider one of the $48$ terms of $m^{(3)}_6$ and compute its residue
	\be  
	\oint_{|\s_{456}|=\epsilon}d \s_{456} \frac{1}{\s_{612}\t_{3456}(\t_{3456}+\s_{134}+\s_{234})\s_{456}} =\frac{1}{\s_{612}\t_{3456}(\t_{3456}+\s_{134}+\s_{234})} =  \frac{1}{s_{23}t_{234}s_{61}},
	\ee 
	where the second equality was obtained by applying the map \eqref{map36to26} from $k=3$ to $k=2$ Mandelstam invariants. Here $t_{234}=s_{23}+s_{24}+s_{34}$.

	\section{Different Paths and an Emergent Amplitude}\label{sec: paths emergent amp}
	
	Up to this point we have discussed a single $(n-5)$-dimensional residue of $m^{(3)}_n$ that gives $m^{(2)}_n$. Of course, any other residue related to \eqref{poles} via a cyclic shift in the indices also leads to $m^{(2)}_n$. However, even after fixing the kinematic invariant with the longest width, e.g., by choosing $\t_{4,5,\ldots ,n}$, we find exactly $C_{n-5}$ different residues. The appearance of the Catalan numbers is not an accident. Each $(n-5)$-dimensional residue can be defined from a triangulation of a $(n-3)$-gon. In our example, the vertices of the polygon are given, in order, by $\{ 4,5,\ldots , n \}$.  
	
	\begin{prop}\label{propA}
		For any given triangulation of the $(n-3)$-gon with vertices $\{ 4,5,\ldots , n \}$, associate a kinematic invariant, $\t_J$, to a diagonal between vertices $a$ and $b$ with $a<b$ by setting $J = \{ a,a+1,\ldots ,b-1,b\}$. Then, the $(n-5)$-dimensional residue of $m^{(3)}_n$ defined by
		$$(\t_{J_0},\t_{J_1},\t_{J_2}, \ldots ,\t_{J_{n-6}}) = (0,0,0,\ldots ,0), $$
		with $J_0=\{4,5,\ldots ,n\}$ and $\t_{J_1},\t_{J_2}, \ldots ,\t_{J_{n-6}}$ obtained from the $n-6$ diagonals, is equal to $m^{(2)}_n$ under some identification of kinematic invariants.
	\end{prop}
	Here we only provide a sketch of a proof. Consider a more general (and redundant) version of the parameterization \eqref{repreM} given by
	\be\label{genM} 
	\left[ \begin{array}{ccccccccc}
		1 & 0 & 0 & 1+w_1 & x_2+w_2 & x_2-x_3+w_3  & x_2-x_4+w_4 & \cdots & x_2-x_{n-3}+w_{n-3}  \\
		0 & 1 & 0 & 1 & x_1 & x_1-x_3  & x_1-x_4 & \cdots & x_1-x_{n-3} \\
		0 & 0 & 1 & 1 &  1  &  1-x_3   & 1-x_4   & \cdots & 1 - x_{n-3} 
	\end{array}  \right] .
	\ee 
	The next step is to compute the $n-5$ planar minors of the matrix, $\Delta_{456},\Delta_{567},\ldots \Delta_{n-2,n-1,n}$, in order to express $w_1,w_2\ldots ,w_{n-3}$ in terms of them. Of course, there are two redundant $w$'s but they can be set to zero to simplify the computation in different ways depending on the residue of interest. In this way 
	\be 
	w_i=w_i(\Delta_{456},\Delta_{567},\ldots ,\Delta_{n-2,n-1,n}).
	\ee 
	Different $(n-5)$-dimensional residues are obtained by selecting the way the planar minors depend on the deformation parameters. For example, the residue studied in section \ref{particularR} is obtained by setting 
	\be  
	\Delta_{456} =z_1,~\Delta_{567}=z_1z_2,~\Delta_{678}=z_1z_2z_3, \ldots ,\Delta_{n-2,n-1,n}=\prod_{i=1}^{n-5}z_i.
	\ee 
	The attentive reader might notice that there are $(n-5)!$ ways of assigning the $z$-dependence while the claim is that there are only $C_{n-5}$ distinct residues. 
	
	The way $(n-5)!$ assignments reduce to $C_{n-5}$ is by considering the way a given pattern implies configurations of points on $\mathbb{CP}^2$ and the corresponding ``dual'' kinematic invariant needed to be set to zero to keep the function ${\cal S}^{(3)}_n$ finite. 
	
	This is best explained with an example. Consider the $n=8$ case and let us study the $(8-5)!=6$ ways of identifying $\Delta_{456},\Delta_{567},\Delta_{678}$ with $z_1,z_1z_2,z_1z_2z_3$. 
	
	Under 
	\be  
	\Delta_{456} =z_1,~\Delta_{567}=z_1z_2,~\Delta_{678}=z_1z_2z_3,
	\ee 
	one sets $\Delta_{678}=0$ first which means that points $6,7,8$ are on a line in $\mathbb{CP}^2$ and therefore one must set $\s_{678}=0$. Next, $\Delta_{567} =0$ but this time imposing that $5,6,7$ are on a line implies that $5,6,7,8$ are on a line and therefore one must set $\t_{5678} =0 $. Finally, $\Delta_{456} = 0$ sets $4$ to be on the line defined by $5$ and $6$ and therefore all five points are collinear and one must set $\t_{45678}=0$.  
	
	Consider now 
	\be  
	\Delta_{456} =z_1z_2,~\Delta_{567}=z_1,~\Delta_{678}=z_1z_2z_3.
	\ee 
	As before, we have $6,7,8$ on a line with $\s_{678}=0$. Next $4,5,6$ must be on a line. Note that this time this does not imply any more conditions and therefore one simply sets $\s_{456} =0$. Finally, requiring $5,6,7$ to be collinear forces the lines $4,5,6$ and $6,7,8$ to coincide and therefore all five points are collinear so one must set $\t_{45678}=0$. 
	
	Now it is obvious that the pattern 
	\be  
	\Delta_{456} =z_1z_2z_3,~\Delta_{567}=z_1,~\Delta_{678}=z_1z_2
	\ee 
	leads to the same residue as above, i.e., $\s_{456}=\s_{678}=\t_{45678}=0$. 
	
	It is also clear that all residues share $\t_{45678}$ as the end point configuration will always have the $n-5$ points collinear in $\mathbb{CP}^2$. 
	
	Note that this provides a surjection from the set of permutations of $(n-5)$ elements onto the set of triangulations of a $(n-3)$-gon.
	
	The reader familiar with the Grassmannian formulation of SYM will notice that for NMHV amplitudes the pattern of $(n-5)$-dimensional residues is also closely related to Catalan numbers. In fact, it is a refinement known as Narayana numbers\footnote{See e.g. {\tt https://en.wikipedia.org/wiki/Narayana_number}} \cite{Arkani-Hamed:2012zlh}, see \cite{oeisNarayana}. For $k=3$, it is known that each BCFW diagram localizes on a $\mathbb{CP}^2$ in twistor space (i.e. $\mathbb{CP}^3$) \cite{unification}. Each individual BCFW-diagram is obtained by setting a series of planar minors (Plucker coordinates of $G(3,n)$) to zero. Each such vanishing minor implies that the corresponding points become collinear. Still in the SYM context, the $(n-5)$ conditions can be supplemented by partially relaxing momentum conservation to impose that one more planar minor becomes zero, leading to a Cachazo-Svrcek-Witten (CSW) \cite{Cachazo:2004kj} localization \cite{Arkani-Hamed:2009pfk}, i.e., the amplitude is computed by configurations where all particles belong to a pair of lines in $\mathbb{CP}^2$.       
	
	\subsection{An Emergent Amplitude}
	
	Recall the final formula obtained in section \ref{particularR}, \eqref{finn}, now written in the notation of Proposition \ref{propA},
	\be
	\left. m^{(3)}_n\right|_{\rm Leading\, order}  =  \frac{1}{\t_{J_0}}\left(\prod_{i=1}^{n-6}\frac{1}{\t_{J_i}}\right)\, m^{(2)}_n .
	\ee
	This same formula applies to all $C_{n-5}$ residues. This means that if we consider the subspace of kinematic invariants where $\s_{abc} \to \epsilon\,  \hat\s_{abc}$ for all $\{a,b,c\}\subset \{4,5,\ldots ,n\}$, then we get a formula for the leading order behaviour of $m^{(3)}_n$ of the form
	\be\label{twoAmps} 
	\left. m^{(3)}_n\right|_{\rm Leading\, order}  =\frac{1}{\epsilon^{n-5}\hat\t_{J_0}}\left(\sum_{r=1}^{C_{n-5}} \prod_{i=1}^{n-6}\frac{1}{\hat\t_{J^{(r)}_i}} \right) \, m^{(2)}_n. 
	\ee 
	Here we have combined all terms assuming that the $m^{(2)}_n$ can be chosen to be the same for all. This is possible since the map found in \eqref{fullmap} is consistent with setting all kinematic invariants of the form $\s_{abc}$ with $\{a,b,c\}\subset \{4,5,\ldots ,n\}$ to zero. 
	
	Given the bijection with triangulations of a $(n-3)$-gon, the factor in \eqref{twoAmps} can be identified with an amplitude, i.e.
	\be 
	\left(\sum_{r=1}^{C_{n-5}} \prod_{i=1}^{n-6}\frac{1}{\hat\t_{J^{(r)}_i}} \right) = m^{(2)}_{n-3}
	\ee 
	for some identification of kinematic invariants. 
	
	This implies that under the simultaneous limit the leading order of the amplitude becomes 
	\be\label{paraQ}  
	\left. m^{(3)}_n\right|_{\rm Leading\, order}  =\frac{1}{\epsilon^{n-5}\hat\t_{J_0}}\, m^{(2)}_{n-3}\, m^{(2)}_n.
	\ee 

	\subsection{Parallel Hard Limit}\label{sec: parallelHardLimit}
	
	The form of \eqref{paraQ} is reminiscent of formulas for the so-called ``soft'' and ``hard'' limits \cite{GarciaSepulveda:2019jxn}. In the former, all kinematic invariants containing a given particle are taken to zero while in the latter those that do not contain the given particle are the ones taken to zero. It was found in \cite{GarciaSepulveda:2019jxn} that 
	\be  
	\left. m^{(k)}_n\right|_{\rm soft}  = m^{(2)}_{k+2} m^{(k)}_{n-1},\quad \left. m^{(k)}_n\right|_{\rm hard} =  m^{(2)}_{n-k+2} m^{(k-1)}_{n-1}.
	\ee 
	In these formulas the first factor is recognized as an amplitude only after an appropriate identification of kinematic invariants while the second factor is straightforwardly an amplitude with kinematics inherited from the original amplitude via the corresponding limit. 
	
	As it turns out, it is possible to associate the kinematics leading to \eqref{paraQ}, i.e. 
	\be  
	\left. m^{(3)}_n\right|_{\rm Leading\, order}  =\frac{1}{\epsilon^{n-5}\hat\t_{J_0}}\, m^{(2)}_{n-3}\, m^{(2)}_n.
	\ee 
	with what can be called a {\it parallel hard limit} on particles $\{1,2,3\}$ or, equivalently, a {\it parallel soft limit} on particles $\{ 4,5,\ldots ,n\}$. This equivalence has to do with the ``parallel'' nature of the limit which we now explain. 
	
	For convenience, we choose the version with the smaller number of particles involved, i.e., the parallel hard limit: Consider the limit in which every kinematic invariant of the form $\s_{1ab},\s_{2ab},\s_{3ab}$, with $a,b\in [n]$, is taken to be large, of order $1/\epsilon$. It is important note that this is not a multiple-hard limit, which is defined by $\s_{abc}\sim \epsilon^{-|\{a,b,c\}\cap \{1,2,3\}|} $. For example, in a multiple-hard limit $\s_{124}\sim \epsilon^{-2}$ while in the parallel hard limit $\s_{124}\sim \epsilon^{-1}$. 
	
	Note that, up to an irrelevant overall factor, this limit is easily seen to be equivalent to the one in which every kinematic invariant of the form $\s_{abc}$ with $\{ a,b,c\}\subset \{4,5,\ldots ,n\}$ is small, i.e. of order $\epsilon$ while others are kept finite. Once again, this is different from a multiple-soft limit (see e.g. \cite{Abhishek:2020xfy,Abhishek:2020sdr}) on the same set of particles as, for example, $\s_{124}$ would also be required to vanish while in the parallel soft limit it remains finite.

	In the next section we start the exploration showing how connections to the hypersimplex and to the planar basis, introduced by the second author in \cite{Early:2019eun}, lead to a proposal, aimed toward classifying the $(k-2)(n-k-2)$-dimensional residues that agree with the physical amplitude, $m^{(2)}_{n}$.

	\section{Generalization to $k\ge 3$: From Residues to Associahedra}\label{sec: generalization higher k proposal}

	In this section, we formulate a conjecture for the identification of $m^{(2)}_n$ with a particular $(k-2)(n-k-2)$-dimensional residue of $m^{(k)}_n$, for any $3\le k\le n-3$.
	
	In order to explain our main formula, that is the generalization of Equation \eqref{finn} to all $k\ge 3$, it is convenient to recall the basis of planar kinematic invariants \cite{Early:2019eun}. Given a subset $J\subseteq \{1,\ldots, n\}$, define $x_J = \sum_{j\in J}x_j$ and $e_J = \sum_{j\in J} e_j$.  Here $\{e_1,\ldots, e_n\}$ is the standard basis for $\mathbb{R}^n$. 
	
	Let $\binom{\lbrack n\rbrack}{k}$ be the set of all $k$-element subsets of $\{1,\ldots, n\}$ and let $\binomial{\lbrack n\rbrack}{k}^{nf}$ be the set of all $k$-element \textit{nonfrozen} subsets, that is, those that are not cyclic intervals modulo $n$.  
	
	The kinematic space $\mathcal{K}(k,n)$ is a codimension $n$ subspace of $\mathbb{R}^{\binom{n}{k}}$, given by 
	\begin{eqnarray}
		\mathcal{K}(k,n) & = & \left\{(\mathfrak{s}) \in \mathbb{R}^{\binom{n}{k}}: \sum_{J\ni j}\s_J = 0,\ j=1,\ldots, n \right\}.
	\end{eqnarray}
	
	\subsection{Multi-Dimensional Residues at $k\ge 4$}\label{pima}
	We would like to find an identification of $m^{(2)}_n$ with some residues of $m^{(k)}_n$, for any $3\le k\le n-3$.  These residues should be $(k-2)(n-k-2)$-dimensional.  However, the analog of the base kinematic invariants $\mathbf{t}_{45\cdots n}$, see Proposition \ref{propA}, involves more complicated as the base needs to include kinematic invariants that are not of the simple form $\mathfrak{t}_{J}$.  In this section, we recall a construction of a basis of kinematic invariants that solves our problem and gives a systematic construction of the base kinematic invariants, with an explicit formula.
	
	\begin{defn}[\cite{Early:2019eun}]
		For any $J\in \binom{\lbrack n\rbrack}{k}^{nf}$, define a linear function on the kinematic space $\eta_J:\mathcal{K}(k,n) \rightarrow \mathbb{R}$, by
		\begin{eqnarray}\label{eq:planar basis element}
			\eta_J(\mathfrak{s}) & = & -\frac{1}{n}\sum_{I\in \binom{\lbrack n\rbrack}{k}}\min\{L_1(e_I-e_J),\ldots, L_n(e_I-e_J)\}\mathfrak{s}_I,
		\end{eqnarray}
		where 
		\begin{eqnarray}\label{eq: L's}
			L_j(x) & = & x_{j+1}+2x_{j+2} + \cdots +(n-1)x_{j-1},
		\end{eqnarray}
		for $j=1,\ldots, n$, are linear functions on $\mathbb{R}^n$.
	\end{defn}

	For $k\ge 4$ the first step is to construct an analog of the base kinematic invariant $\mathfrak{t}_{J_0}$.  For motivation, consider that there is a well-known duality $m^{(k)}_n = m^{(n-k)}_{n}$, and after the following identification of kinematic invariants: for any $I = \{i_1,\ldots, i_k\}$, then one identifies the Mandelstam invariant $\mathfrak{s}_I$ on $\mathcal{K}(k,n)$ with the Mandelstam invariant $\mathfrak{s}_{I^c}$ on $\mathcal{K}(n-k,n)$.
	
	This allows us to construct the base kinematic invariants for $m^{(n-3)}_n$ from the full set of kinematic invariants of $m^{(3)}_n$.  In the planar basis, the kinematic invariants in Equation \eqref{poles} become
	\begin{eqnarray*}
		(\t_{45\ldots n},\ \t_{56\ldots n},\ \t_{67\ldots n},\ldots, \s_{n-2,n-1,n}) & = & (\eta_{3,n-1,n},\eta_{4,n-1,n},\eta_{5,n-1,n},\ldots, \eta_{n-3,n-1,n}).
	\end{eqnarray*}
	According to the general theory developed in \cite{Nick-forthcoming}, any planar basis element of the form $\eta_{j,n-1,n}$ dualizes under the identification $\mathfrak{s}_J \mapsto \mathfrak{s}_{J^c}$ to $\eta_{\{2,\ldots, j\} \cup \{n-j+1,n\}}$.  Then for example
	$$(\eta_{3,6,7},\eta_{4,6,7})$$
	dualizes to 
	$$(\eta_{2,3,6,7},\eta_{2,3,4,7}).$$
	To bring this to the standard form, after cyclically permuting with $i\mapsto i+3\mod(7)$ we obtain
	$$(\eta_{3,4,6,7},\eta_{3,5,6,7}).$$
	These become the base kinematic invariants for (k,n) = (4,7) and for any (4,n), there are two base kinematic invariants
	$$(\eta_{3,4,n-1,n},\eta_{3,n-2,n-1,n}).$$
	In general there are $k-2$ of them.  
	
	For our ansatz, we now define $k-2$ planar kinematic invariants to serve as base poles,
	\begin{eqnarray}\label{eq: base kinematic invariants}
		\text{base}_{k,n} & = & \left\{\eta_{\lbrack 3,i\rbrack\cup \lbrack n-(k-i+1),n\rbrack }: i=3,\ldots, k\right\}.
	\end{eqnarray}
	For instance, 
	$$\text{base}_{3,6} = \{\eta_{3,5,6} \},\ \text{base}_{3,7} = \{\eta_{3,6,7} \},\ \text{base}_{4,7} = \{\eta_{3,5,6,7},\eta_{3,4,6,7} \},$$
	$$\text{base}_{5,10} = \{\eta_{3,7,8,9,10},\eta_{3,4,8,9,10},\eta_{3,4,5,9,10} \}.$$
	Once the base kinematic invariants are given, in order to arrive at an identification of a residue of $m^{(k)}_n$ with $m^{(2)}_n$ one still has to take a dimension $(k-2)(n-k-2)-(k-2)$ residue.  But which?  Certainly not every such residue can be identified with $m^{(2)}_n$, so a priori choosing the residue seems like an enormously challenging problem, due to the complexity of $m^{(k)}_n$.  However, we nonetheless have enough data to formulate an all $(k,n)$ conjecture; we will return to this in \cite{future}.
	
	In order to state the conjecture we need an additional $(k-2)(n-k-2) - (k-2)$ planar kinematic invariants,
	\begin{eqnarray}
		\text{comb}_{k,n} & = &  \left\{ \eta_{\lbrack j,j+r\rbrack \cup \lbrack n-r,n\rbrack}: j=3,\ldots, n-k,\ \text{ and } r=1,\ldots, k-2 \right\}.
	\end{eqnarray}
	Here for convenience we are using the notation $\text{comb}_{k,n}$ to reflect that the collection of index sets is ``combed'' towards $n$.
	
	For example,
	\begin{eqnarray*}
		\text{comb}_{3,9} & = & \{\eta_{489},\eta_{589},\eta_{689}\},\\
		\text{comb}_{4,9} & = & \{\eta_{4789},\eta_{5789},\eta_{4589},\eta_{5689}\},\\
		\text{comb}_{5,9} & = & \{\eta_{46789},\eta_{45789},\eta_{45689}\}.
	\end{eqnarray*}

	\begin{conjecture}\label{conjecture: m2n in mkn}
		The $(k-2)(n-k-2)$-dimensional residue of $m^{(k)}_n$ on the subspace where $\eta_J = 0$ for all $\eta_J \in \text{base}_{k,n}\cup \text{comb}_{k,n}$, can be identified with $m^{(2)}_n$.  That is, there exists an identification of kinematics such that 
		\begin{eqnarray*}
			\text{Res}\lbrack m^{(k)}_n\rbrack_{\eta_J = 0,\ \eta_J \in \text{base}_{k,n}\cup\  \text{comb}_{k,n}} & = & m^{(2)}_n.
		\end{eqnarray*}
	\end{conjecture}

	\subsection{Combinatorial Simplification and Generalization via the Planar Basis}

	Let us revisit the Example in Section \ref{sec: examples} but now in the planar basis, in order to reveal a hidden combinatorial structure.

		The residue of $m^{(3)}_6$ at the pole $\eta_{356} (=\mathfrak{s}_{456})= 0$ is
		\begin{eqnarray*}
			& & \text{Res}\lbrack m^{(3)}_6\rbrack_{\eta_{356} = 0} \\
			& = &\frac{1}{\eta _{134} \eta _{135} \eta _{136}}+\frac{1}{\eta _{125} \eta _{135} \eta _{235}}+\frac{1}{\eta _{135} \eta _{136} \eta _{235}}+\frac{1}{\eta _{136} \eta _{235} \eta _{236}}+\frac{1}{\eta _{125} \eta _{134} \eta _{135}}\\
			& + & \frac{1}{\eta _{235} \eta _{236} \eta _{256}}+\frac{1}{\eta _{134} \eta _{136} \eta _{346}}+\frac{1}{\eta _{136} \eta _{236} \eta _{346}}+\frac{1}{\eta _{236} \eta _{256} \eta _{346}}+\frac{1}{\eta _{125} \eta _{235} \eta _{256}}\\
			& + & \frac{1}{\eta _{125} \eta _{256} \left(\eta _{124}-\eta _{246}+\eta _{256}+\eta _{346}\right)}+\frac{1}{\eta _{125} \eta _{134} \left(\eta _{124}-\eta _{246}+\eta _{256}+\eta _{346}\right)}\\
			& + & \frac{1}{\eta _{256} \eta _{346} \left(\eta _{124}-\eta _{246}+\eta _{256}+\eta _{346}\right)}+\frac{1}{\eta _{134} \eta _{346} \left(\eta _{124}-\eta _{246}+\eta _{256}+\eta _{346}\right)}.
		\end{eqnarray*}
		This has poles at the vanishing of the set of nine planar kinematic invariants, respectively
		\begin{eqnarray}\label{eq:poles residue 36}
			\eta _{125},\eta _{134},\eta _{135},\eta _{136},\eta _{235},\eta _{236},\eta _{256},\eta _{346},\eta _{124}-\eta _{246}+\eta _{256}+\eta _{346}.
		\end{eqnarray}
		In this case, all of the relevant elements of the planar basis have simple expressions, that is
		\begin{eqnarray*}
			\eta_{135} & = & \frac{1}{2} \mathfrak{s}_{123}+\frac{1}{3} \mathfrak{s}_{124}+\frac{1}{6} \mathfrak{s}_{125}+\mathfrak{s}_{126}+\frac{1}{6} \mathfrak{s}_{134}+\frac{5}{6} \mathfrak{s}_{136}+\frac{5}{6} \mathfrak{s}_{145}+\frac{2}{3} \mathfrak{s}_{146}+\frac{1}{2} \mathfrak{s}_{156}+\mathfrak{s}_{234}\\
			& + & \frac{5}{6} \mathfrak{s}_{235}+\frac{2}{3} \mathfrak{s}_{236}+\frac{2}{3} \mathfrak{s}_{245}+\frac{1}{2} \mathfrak{s}_{246}+\frac{1}{3} \mathfrak{s}_{256}+\frac{1}{2} \mathfrak{s}_{345}+\frac{1}{3} \mathfrak{s}_{346}+\frac{1}{6} \mathfrak{s}_{356}+\mathfrak{s}_{456}\\
			& = & \t_{1236} + \mathfrak{s}_{234} + \mathfrak{s}_{235},
		\end{eqnarray*}
		say, and similarly
		\begin{eqnarray*}
			& \eta_{124} = \mathfrak{t}_{1256},\ \eta_{134} = \mathfrak{s}_{234},
		\end{eqnarray*}
		taking into account momentum conservation for the simplification.  The others are obtained by cyclic index relabeling.
		
		Let us point out that something interesting happens if we additionally assume that  $\mathfrak{s}_{356} = 0$.  The identity
		\begin{eqnarray}
			-\mathfrak{s}_{356} & = & \eta_{356} + \eta_{246} - \eta_{346} - \eta_{256}
		\end{eqnarray}
		leads directly to  
		\begin{eqnarray}
			\eta_{246} & = & \eta_{346} + \eta_{256},
		\end{eqnarray}
		having used that $\eta_{356}=0$, and the residue simplifies very nicely to 
		\begin{eqnarray}\label{eq: NC Equivalence}
			& & \left(\text{Res}\lbrack m^{(3)}_6\rbrack_{\eta_{356}=0}\right)\big\vert_{\mathfrak{s}_{356} = 0} \nonumber\\
			& = & \frac{1}{\eta _{125} \eta _{134} \eta _{135}}+\frac{1}{\eta _{134} \eta _{135} \eta _{136}}+\frac{1}{\eta _{125} \eta _{135} \eta _{235}}+\frac{1}{\eta _{135} \eta _{136} \eta _{235}}+\frac{1}{\eta _{136} \eta _{235} \eta _{236}}\nonumber\\
			& + & \frac{1}{\eta _{124} \eta _{125} \eta _{134}}+\frac{1}{\eta _{125} \eta _{235} \eta _{256}}+\frac{1}{\eta _{235} \eta _{236} \eta _{256}}+\frac{1}{\eta _{124} \eta _{134} \eta _{346}}+\frac{1}{\eta _{134} \eta _{136} \eta _{346}}\\
			& + & \frac{1}{\eta _{136} \eta _{236} \eta _{346}}+\frac{1}{\eta _{124} \eta _{125} \eta _{256}} + \frac{1}{\eta _{236} \eta _{256} \eta _{346}}+\frac{1}{\eta _{124} \eta _{256} \eta _{346}},\nonumber
		\end{eqnarray}
		where we emphasize that the nine planar kinematic invariants appearing above are still linearly independent even after the additional constraint $\s_{356}=0$.  Moreover, one can check that the 14 fractions are naturally in bijection with the maximal pairwise weakly separated collections of the corresponding triples,
		$$\{1,2,4\},\{1,2,5\},\{1,3,4\},\{1,3,5\},\{1,3,6\},\{2,3,5\},\{2,3,6\},\{2,5,6\},\{3,4,6\}.$$
		Is an analogous simplification to Equation \eqref{eq: NC Equivalence} always possible for $m^{(3)}_n$ for all $n\ge 6$ and any of the residues considered above? Indeed, we have found analogous statements for residues of $m^{(3)}_n$ with $n=6,7,8$, except that for $n\ge 7$ the pairwise weak separation criterion for pole compatibility had to be replaced with the noncrossing condition, see for instance \cite{santos2017noncrossing}; in the context of generalized amplitudes see \cite{Early:2021tce}.  
		
		A noncrossing analog $m^{(k,NC)}_n$ of $m^{(k)}_n$ was proposed in \cite{Early:2021tce}; in fact two definitions were given.  One involves a CEGM-like construction with the scattering equations.  The conjecture is that they equivalent.  The other is purely combinatorial as it involves a sum over all maximal pairwise noncrossing collections in $\binom{\lbrack n\rbrack}{k}^{nf}$.  Let us recall briefly recall the latter in the case $k=3$.
		
		A pair $\{\{i_1,j_1,k_1\},\{i_2,j_2,k_2\}\}$ in $\binom{\lbrack n\rbrack}{3}^{nf}$ is called \textit{noncrossing} provided that none of the following conditions hold:
		\begin{enumerate}
			\item $i_1<i_2<j_1<j_2$, or $i_2<i_1<j_2<j_1$,
			\item $j_1<k_2<j_1<k_2$, or $j_2<k_1<j_2<k_1$,
			\item $j_1 = j_2$, and either $i_1<i_2<k_1<k_2$ or $i_2<i_1<k_2<j_1$.
		\end{enumerate}
		For example, according to (3) the pairs 
		$$\{145,236\},\{124,356\} \in \mathbf{NC}_{3,6}$$
		are noncrossing since the middle elements are not the same, respectively $4\not=3$ and $2\not=5$.  It is well-known \cite{santos2017noncrossing} (and in any case easy to see from the definition) that for $k\ge 3$ the noncrossing complex $\mathbf{NC}_{k,n}$ is not cyclically invariant.

		Denote by $\mathbf{NC}_{3,n}$ the collection of all pairwise noncrossing collections of elements of $\binom{\lbrack n\rbrack}{3}^{nf}$.  For general results concerning $\mathbf{NC}_{k,n}$, we refer the reader to \cite{santos2017noncrossing}.  It is known that there are $C^{(3)}_{n-3} = 5,42,462,6006,\ldots $ such maximal (by inclusion) collections for $n=5,6,7,8,\ldots$, where $C^{(3)}_{n-3}$ denotes the 3-dimensional Catalan numbers.  Moreover, every such collection has exactly $(2)(n-4)$ elements, excluding frozen subsets.  Note that in this scheme we recover the 2-dimensional Catalan numbers $C^{(2)}_{n-2} = 2,5,14,42,132,\ldots$ for $n=4,5,6,7,8,\ldots$, as a special case of the $k$-dimensional Catalan numbers $C^{(k)}_{n-k}$, \cite{oeisdimkCatalan}.
		
		Theorem \ref{thm: noncrossing k=3 reduction} concerns $n=6,7,8$; a direct computation of the amplitude for $n\ge 9$ is not feasible using existing techniques; nonetheless we conjecture that the analogous result holds for all $n\ge 6$.  
		\begin{thm}\label{thm: noncrossing k=3 reduction}
			For $n=6,7,8$, given any (maximal) noncrossing collection 
			$$\{(i_1,j_1),\ldots, (i_{n-6},j_{n-6})\} \in \mathbf{NC}_{2,n-3},$$
			then there exists a subspace $\mathcal{K}_0$ of the kinematic space such that
			\begin{eqnarray*}
				\frac{1}{\eta_{3,n-1,n}}\left(\prod_{r=1}^{n-6}\frac{1}{\eta_{i_r+2,j_r+2,j_r+3}}\right)\left(\text{Res}\lbrack m^{(3)}_n\rbrack_{(\eta_{3,n-1,n},\eta_{i_1+2,j_1+2,j_1+3},\ldots) = (0,\ldots, 0)}\right)\big\vert_{\mathcal{K}_0} & = & \sum_{\mathbf{J}}\prod_{J\in \mathbf{J}}\frac{1}{\eta_{J}},
			\end{eqnarray*}
			Here the sum is over all maximal collections $\mathbf{J} = \{J_1,\ldots, J_{2(n-4)}\} \in \mathbf{NC}_{3,n}$ satisfying 
			$$\mathbf{J}\supset \{\{3,n-1,n\},\{i_1+2,j_1+2,j_1+3\},\ldots, \{i_{n-6}+2,j_{n-6}+2,j_{n-6}+3\}\}.$$

		\end{thm}

		\begin{proof}
			The case $n=6$ was given explicitly in Equation \eqref{eq: NC Equivalence}.
			
			The explicit calculation for $n=7$ is included in an ancillary file to the arXiv version of the paper.  For $n=7$, and the residue where 
			$$(\eta_{367},\eta_{356}) = (0,0)$$
			if we additionally impose the linear relations 
			$$\eta_{257} = \eta_{267} + \eta_{357} \text{ and } \eta_{247} = \eta_{267} + \eta_{347}$$
			then the result holds.
			
			As for the residue where 
			$$(\eta_{367},\eta_{467}) = (0,0),$$
			in addition to 
			$$\eta_{257} = \eta_{267} + \eta_{357} \text{ and } \eta_{247} = \eta_{267} + \eta_{347}$$
			one also has to set
			$$\eta_{357} = \eta_{457}.$$
			For $(3,8)$ and the residue where 
			$$(\eta_{378},\eta_{478},\eta_{578}) = (0,0,0),$$
			if we additionally impose the following conditions
			$$\eta_{358} = \eta_{458},\ \eta_{368} = \eta_{568},\ \eta_{468} = \eta_{568},\ \eta_{248} = \eta_{278}+\eta_{348},\ \eta_{258} = \eta_{278} + \eta_{458},\ \eta_{268} = \eta_{278} + \eta_{568},$$
			then the result follows.  The remaining four residues are similar.
			%
		\end{proof}

		There seems to be a quite analogous story for all $k\ge 4$; details of the general construction are deferred to \cite{future}.

\subsection{Associahedra From Positroid Polytopes}
In this section, we use what we have learned from Theorem \ref{residual3} about the identification of residues with $m^{(2)}_n$ to make a combinatorial digression.

In the computation of the residues considered in this paper, the positive parametrization was not used; but it could have been.  In this case one encounters a triangle of new realizations of associahedra, the $n^\text{th}$ row ranges over $k=2,3,\ldots, n-2$, giving the $(n-3)$-dimensional associahedron.

In order to be consistent with Theorem \ref{residual3}, the Newton polytopes below that we construct should be combinatorially isomorphic to $(n-3)$ dimensional associahedra.

The polytopes under consideration are the Newton polytopes
\begin{eqnarray}\label{eq: standard rank 3 embedding associahedron}
	\text{Newt}\left((x_{1,1}+x_{1,2})\prod_{j=2}^{n-3}\delta_{(\lbrack 1,2\rbrack,\lbrack 1,j\rbrack)}\prod_{1\le a<b\le n-3}(x_{2,a} +x_{2,a+1}+ \cdots x_{2,b})\right),
\end{eqnarray}
where we denote
$$\delta_{(\lbrack 1,2\rbrack,\lbrack 1,j\rbrack)} = x_{1,1}(x_{2,1} + \cdots + x_{2,j}) + x_{1,2}(x_{2,2} + \cdots + x_{2,j}).$$
Now after setting
$$x_{1,1} = y_{1},\ x_{1,2} = y_{2},\text{ and } \ x_{2,1} = y_2,\ x_{2,2} = y_3,\ldots,\  x_{2,n-3} = y_{n-2},$$
it is not difficult to verify that the Newton polytope in Equation \eqref{eq: standard rank 3 embedding associahedron} projects isomorphically to a Minkowski sum of an n-4 dimensional associahedron in Loday's realization \cite{loday2004realization} as a generalized permutohedron \cite{postnikov2009permutohedra}, together with a line segment and a collection of n-4 positroid polytopes in the second hypersimplices $\Delta_{2,j+1}$, that is, after the change of variable, then $$\text{Newt}(\delta_{(\lbrack 1,2\rbrack,\lbrack 1,j\rbrack)})$$ becomes the positroid polytope,
\begin{eqnarray*}
	\left\{x\in \Delta_{2,j+1}: x_1+x_2 \ge  1 \right\},
\end{eqnarray*}
for $j=2,\ldots, n-2$.

For example, for $(k,n) = (3,6)$ we find the Newton polytope
$$
\begin{array}{c}
	x_{1,1}+x_{1,2} \\
	x_{2,1}+x_{2,2} \\
	x_{1,1} x_{2,1}+x_{1,1} x_{2,2}+x_{1,2} x_{2,2} \\
	x_{2,2}+x_{2,3} \\
	x_{2,1}+x_{2,2}+x_{2,3} \\
	x_{1,1} x_{2,1}+x_{1,1} x_{2,2}+x_{1,2} x_{2,2}+x_{1,1} x_{2,3}+x_{1,2} x_{2,3} \\
\end{array}
$$
and after the substitution the Newton polytope of the product of 
$$
\begin{array}{c}
	y_1+y_2 \\
	y_2+y_3 \\
	y_1 y_2+y_3 y_2+y_1 y_3 \\
	y_3+y_4 \\
	y_2+y_3+y_4 \\
	y_1 y_2+y_3 y_2+y_4 y_2+y_1 y_3+y_1 y_4,
\end{array}
$$
where the Newton polytope of the last line is the half-octahedron,
$$\text{Newt}(y_1 y_2+y_3 y_2+y_4 y_2+y_1 y_3+y_1 y_4) = \left\{x\in \Delta_{2,4}: x_1+x_2\ge 1\right\}.$$
As the Newton polytope in Equation \eqref{eq: standard rank 3 embedding associahedron} is a Minkowski sum of matroid polytopes it is a generalized permutohedron, but it is clearly not in the deformation cone of the associahedron as one of the facet inequalities has the form $x_1+x_3\ge 1$; on the other hand, the facet inequalities of the associahedron in the realization, as a generalized permutohedron, consist of only linear intervals of the form $\sum_{j=a}^bx_j\ge c$.

It is the presence of these positroid polytopes in the realization of the associahedron which could be interesting from a combinatorial point of view.  A similar story, involving positroid polytopes in hypersimplices $\Delta_{k,n}$, appears to hold for $k\ge 4$ as far as we have been able to check; now the implications remain to be determined, but at any rate, to the best of our knowledge, the structures are new (compare with for instance \cite{ceballos2015many}) -- and giving a combinatorial proof for our conjecture, that the $k=3$ Newton polytopes in Equation \eqref{eq: standard rank 3 embedding associahedron} are combinatorially isomorphic (particularly after the projection) to associahedra for all $n\ge 6$ seems like an interesting starting point for investigation.

\begin{figure}
	\centering
	\includegraphics[width=0.6\linewidth]{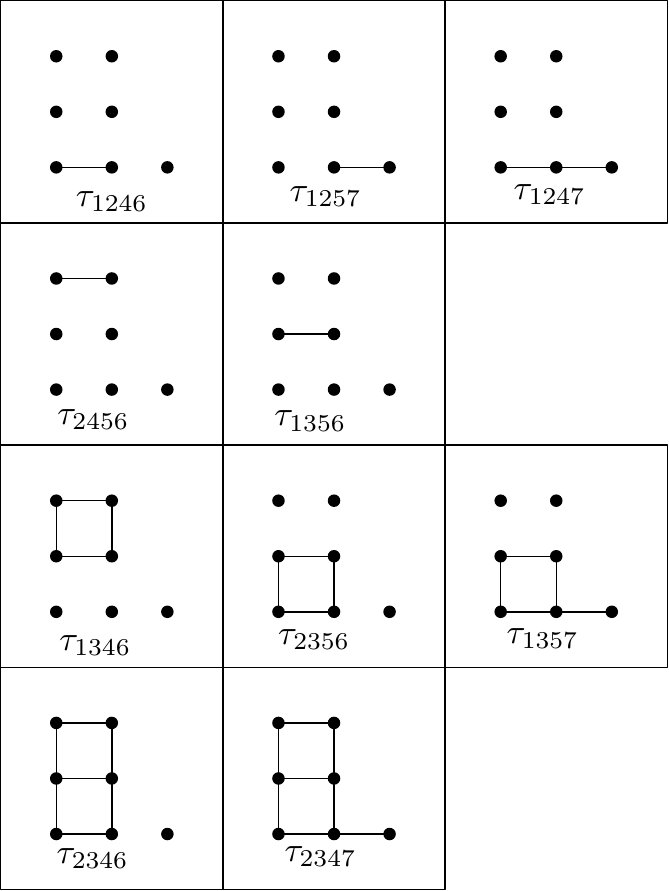}
	\caption{Newton polytope $\text{Assoc}^\ast(4,7)$.  This is combinatorially isomorphic to the $4$-dimensional associahedron.  Take all $\tau_J$'s such that their support lies within the ``hook''.  The hook as $k-2$ rows of length $2$ and one row of length $n-k$.  For details of the diagrammatic construction and the Definition of the polynomials $\tau_J$, see \cite{Early:2021tce}.}
	\label{fig:associahedron47}
\end{figure}

The $n\cdot C_{n-5}$ residues that we have considered are in bijection with certain faces of the Newton polytope 
$$\text{Newt}\left(\prod_{ijk}\Delta_{ijk}\right),$$
having chosen some positive parametrization.

One can apply the same methods used above to derive Equation \eqref{eq: standard rank 3 embedding associahedron} to any $k\ge 4$ positive parametrization and then extrapolate to get a whole hierarchy of novel realizations of the associahedron: for each $n$ there will be $n-2$ realizations of the $n-3$ dimensional associahedron, where the usual realization occurs when $k=2$.  So after applying the $k\ge 4$ analog of projection above, that is
$$x_{1,j} = y_j,\ x_{2,j} = y_{j+1},\ldots, x_{k-2,j} = y_{j+k-2}$$
for $j=1,2$, while for $j=1,\ldots, n-k$
$$x_{k-1,j} = y_{j+k-2},$$
then for a given $(k,n)$, the Minkowski summands will involve positroid polytopes of ranks $j=1,\ldots, k-1$.

Data is included as an ancillary file attached to the arXiv version of the paper.

Figure \ref{fig:associahedron47} contains the combinatorial data needed to reconstruct the $k=4$ realization $\text{Assoc}^\ast(4,7)$ of the 4-dimensional associahedron.  The star indicates that it is still conjectural that for general (k,n) the face lattice of $\text{Assoc}^\ast(k,n)$ is the same as the face lattice of the standard $n-3$-dimensional associahedron.  The explicit polynomials can be found (after running the Definitions component) by evaluating ``assocStar4n[7]'' in the attached notebook.  Also included are explicit sets of polynomials for $(k,n)$ in  
$$\{(3,n):n=5,\ldots, 10\},\ \{(4,n):n=6,\ldots, 10\},\ \{(5,n):n=7,\ldots, 10\}.$$

\section{Veronese Generalization of Biadjoint Amplitudes}\label{veronese}

Continuing with the analogy between $(k,n)$ SYM amplitude formulations and $(k,n)$ generalized biadjoint amplitudes, the next step is to find an analog of the Witten-RSV formulation \cite{Witten:2003nn,Roiban:2004yf} or in its more useful form using the Veronese embedding as presented in \cite{unification}. Here we only scratch the surface by proving a formula for $(3,6)$ but we suspect that general $(k,n)$ formulas should follow from the corresponding formulas in \cite{unification}. In fact this is part of the larger theme of Veronese hypersurface arrangements in the Grassmannian $G(3,n)$ introduced and studied in \cite{Nick-forthcoming}.  

Consider the following Veronese subvariety of $X(3,6)$,
\be\label{vero6} 
M := \left[ \begin{array}{cccccc}
	x_1^2 & x_2^2 & x_3^2 & x_4^2 &  x_5^2 & x_6^2  \\
	x_1 & x_2 & x_3 & x_4 & x_5 & x_6  \\
	1 & 1 & 1 & 1 &  1  &  1
\end{array}  \right] .
\ee 
This can now be brought to a canonical frame and then deformed so that it provides a redundant parameterization of $X(3,6)$, 
\be\label{veroDef} 
\left[ \begin{array}{cccccc}
	1 & 0 & 0 & 1+z &  \frac{(x_1-x_4)(x_3-x_5)}{(x_3-x_4)(x_1-x_5)} &\frac{(x_1-x_4)(x_3-x_6)}{(x_3-x_4)(x_1-x_6)}  \\
	0 & 1 & 0 & 1 & \frac{(x_2-x_4)(x_3-x_5)}{(x_3-x_4)(x_2-x_5)} &\frac{(x_2-x_4)(x_3-x_6)}{(x_3-x_4)(x_2-x_6)}  \\
	0 & 0 & 1 & 1 &  1  &  1
\end{array}  \right] .
\ee 
We choose $\{ x_4,x_5,x_6,z\}$ as variables and keep $\{x_1,x_2,x_3\}$ fixed. 

Now it is clear that the same procedure as in section 3 does not work since the function ${\cal S}^{(3)}_n$ does not develop any singularities as $z\to 0$. This is why one has to deform it as follows,
\be\label{sixV} 
{\cal S}^{(3):V}_6 := \sum_{a<b<c}\s_{abc} \log \Delta_{abc} + \v \log \left(\frac{\Delta_{123}\Delta_{345}\Delta_{561}\Delta_{246}}{\Delta_{234}\Delta_{456}\Delta_{612}\Delta_{135}} - 1\right).
\ee 
Note that the new term is invariant under the torus action and therefore $\v$ does not modify any of the conditions on $\s_{abc}$. We have added a superscript $V$ to the function in order to indicate that it has been modified.

Using \eqref{veroDef} it is easy to see that 
\be 
{\cal S}^{(3):V}_6 = \v \log z + \ldots 
\ee 
where the ellipses indicate terms that remain finite as $z\to 0$.

Following the same steps as in section 3 we find the analog to the equations \eqref{GxGz}
\be 
G_x=\left\{\frac{\partial {\cal S}^{(3):V}_6}{\partial x_i} =0,\quad i\in \{4,5,6\}\right\},\quad G_z=\left\{\frac{\partial {\cal S}^{(3):V}_6}{\partial z} =0 \right\}.
\ee
The equations in the set $G_x$ remain finite as $\v = 0$ and $z=0$. Moreover, they become the scattering equations of a $k=2$ function ${\cal S}^{(2)}_6$ which will be identified after the integrand is constructed. Letting $z=\v q$, the equation in $G_z$ is of the form
\be 
\frac{1}{q} + F(x) = 0
\ee 
and hence there is a single solution for each of the six solutions to $G_x$.

The Jacobian ${\rm det}'\Phi^{(3)}$ is easily computed, since the Hessian matrix is again block diagonal, to be  ${\rm det}'\Phi^{(3)} = ({\rm det}'\Phi^{(2)})/(\v q^2)$.

In order to get a nonzero answer for the amplitude the integrand must develop a singularity when $z=0$. It is easy to check that the Parke-Taylor function ${\rm PT}^{(3)}$ is not singular and therefore it has to be modified. We use the SYM construction as inspiration and define
\be\label{firstA} 
{\rm PT}^{(3):V}_6:= \frac{\Delta_{135}}{\Delta_{123}\Delta_{345}\Delta_{561}(\Delta_{234}\Delta_{456}\Delta_{612}\Delta_{135}-\Delta_{123}\Delta_{345}\Delta_{561}\Delta_{246})}.
\ee
This new function has a simple pole at $z=0$. 

An even more tantalizing way of writing \eqref{firstA} is the following
\be  
{\rm PT}^{(3):V}_6 :=\left( \frac{1}{\Delta_{123}\Delta_{234}\Delta_{345}\Delta_{456}\Delta_{561}\Delta_{612}}\right)\frac{1}{\left(1-\frac{\Delta_{123}\Delta_{345}\Delta_{561}\Delta_{246}}{\Delta_{234}\Delta_{456}\Delta_{612}\Delta_{135}}\right)}.
\ee
This is nothing by the un-deformed Parke-Taylor function times the new term added to ${\cal S}^{(3)}_n$.

We are now led to the following definition.
\begin{defn}\label{defn: vena}
	The generalized Veronese scalar amplitude is given by the following CEGM-like formula
	\be  
	m^{(3):V}_6 := \int \prod_{i=1}^2\prod_{a=4}^6 dy_{ia}\delta\left(\frac{\partial {\cal S}^{(3):V}}{\partial y_{ia}}\right) \left({\rm PT}^{(3):V}_6\right)^2 = \sum_{\rm solutions} \frac{1}{{\rm det}'\Phi^{(3)}} \left({\rm PT}^{(3):V}_6\right)^2.
	\ee 
	Here we have used the standard parameterization of $X(3,6)$ given in \eqref{cegm3para}.
\end{defn}

Computing the precise form of the jacobians to go from the $y_{ia}$ variables to $\{ x_4,x_5,x_6,z\}$ is a simple exercise and combining that with of ${\rm PT}^{(3):V}$ one finds that $m^{(3):V}_6$ has a simple pole at $\v =0$ and its residue agrees with the CHY formula for $m^{(2)}_6$, i.e.
\be 
{\rm Res}[m^{(3):V}_6]_{(\v =0)} = m^{(2)}_6. 
\ee 

The next step is to identify the kinematics invariants in $m^{(2)}_6$. This is easily done by setting $\v = 0$ and $z = 0$ in ${\cal S}^{(3):V}$. Using the parameterization \eqref{vero6} one finds that $\Delta_{abc} =\Delta_{ab}\Delta_{ac}\Delta_{bc}$. Therefore 
\be 
\left. {\cal S}^{(3):V}\right|_{\v=0,z=0} = \sum_{a<b} \left(\sum_c \s_{abc}\right)\log \Delta_{ab}. 
\ee 
Comparing with 
\be
{\cal S}^{(2)} = \sum_{a<b} s_{ab} \log \Delta_{ab},
\ee
one finds a very simply identification
\be 
s_{ab} := \left(\sum_c \s_{abc}\right) = \t_{[n]\setminus \{a,b\}}.
\ee 

Let us explain why this is a very natural formula. In \cite{Cachazo:2019ngv}, CEGM introduced a matrix analog to the helicity formalism for $k>2$ kinematic spaces. For $k=3$ one associates a rank one $3\times 3$ matrix $\mathbb{K}_a$ to each particle so that 
\be  
\s_{abc} = {\rm det}\left( \mathbb{K}_a+\mathbb{K}_b+\mathbb{K}_c\right)
\ee 
while $k=3$ momentum conservation becomes the condition
\be 
{\rm rank}\left( \sum_{a=1}^6 \mathbb{K}_a \right) = 1. 
\ee
Let us rewrite this as
\be 
\sum_{a=1}^6 \mathbb{K}_a = \Lambda \otimes \tilde\Lambda
\ee 
where $\Lambda, \tilde\Lambda \in \mathbb{C}^3$. Using the Cauchy-Binet identity one finds that
\be  
\t_{[n]\setminus \{a,b\}} = {\rm det} \left( \sum_{c\neq a,b} \mathbb{K}_c \right) = -{\rm det} \left( \mathbb{K}_a + \mathbb{K}_b -\Lambda \otimes \tilde\Lambda \right) .
\ee 
We leave the exploration of this natural embedding of ${\cal K}(2,n)$ into ${\cal K}(k,n)$ for future work. 

Before ending this section let us give one more example of how to translate a Grassmannian Veronese formula into the context at hand. Consider the formula given in eq. 3.20 of \cite{unification} for $k=3,n=7$, and identify it with a deformed Parke-Taylor function 
\be\label{transl}  
{\rm PT}^{(3):V}_7 :=\frac{\Delta_{135}\Delta_{612}\Delta_{136}\Delta_{235}}{\Delta_{671}\Delta_{123}\Delta_{345}}\times \frac{1}{\V_{123456}\V_{123567}}.
\ee 
Here the Veronese polynomials are defined so that
$$\V_{123456}:=
\Delta_{234}\Delta_{456}\Delta_{612}\Delta_{135}-\Delta_{123}\Delta_{345}\Delta_{561}\Delta_{246}.$$
Rewriting \eqref{transl} we find 
\be 
{\rm PT}^{(3):V}_7 = {\rm PT}^{(3)}\times \frac{1}{\left(1-\frac{\Delta_{123}\Delta_{345}\Delta_{561}\Delta_{246}}{\Delta_{234}\Delta_{456}\Delta_{612}\Delta_{135}}\right)\left(1-\frac{\Delta_{671}\Delta_{123}\Delta_{356}\Delta_{725}}{\Delta_{567}\Delta_{712}\Delta_{235}\Delta_{613}}\right)} .
\ee 
This form immediately implies that one should define the analog of \eqref{sixV} as 
\be  
{\cal S}^{(3):V}_7\!\!\! := \!\!\!\! \sum_{a<b<c}\!\!\s_{abc} \log \Delta_{abc} + \v_7 \log \left(\frac{\Delta_{123}\Delta_{345}\Delta_{561}\Delta_{246}}{\Delta_{234}\Delta_{456}\Delta_{612}\Delta_{135}} - 1\!\right)+\v_4 \log \left(\frac{\Delta_{671}\Delta_{123}\Delta_{356}\Delta_{725}}{\Delta_{567}\Delta_{712}\Delta_{235}\Delta_{613}}-1\!\right).
\ee 
Once again, one can prove that the generalized Veronese amplitude $m_7^{(3):V}$ has a two-dimensional residue of the form
\be 
{\rm Res}[m^{(3):V}_7]_{(\v_4,\v_7) =(0,0)} = m^{(2)}_7. 
\ee 
We leave it as an exercise to the reader to prove this by using the parameterization of $X(3,7)$ given by
\be\label{vero7} 
M := \left[ \begin{array}{ccccccc}
	x_1^2 & x_2^2 & x_3^2 & x_4^2 &  x_5^2 & x_6^2 & x_7^2 \\
	x_1 & x_2 & x_3 & x_4 & x_5 & x_6 & x_7 \\
	1 & 1 & 1 & 1+z_4 &  1  &  1 & 1+z_7
\end{array}  \right] .
\ee 
Note that we have introduced the coordinates $z_4$ and $z_7$ so that $\V_{123456}\sim z_4, \V_{123567}\sim z_7$. 

Let us end this section by pointing out that explicit forms of $m^{(3):V}_6$ or $m^{(3):V}_7$ in terms of kinematic invariants are not known and it would be interesting to explicitly compute them.

\section{Discussions}\label{sec: discussions}

In this paper, we have found a forest of $n\cdot C^{(2)}_{n-5}$ distinct identifications of the biadjoint scalar $m^{(2)}_n$ with multi-dimensional residues of $m^{(3)}_n$.  Each identification comes with its own parametrization of $X(3,n)$ which is optimally configured for the CEGM scattering equations formula for $m^{(3)}_n$.  After a further kinematic degeneration on three consecutive indices $i,i+1,i+2$, which we call a parallel hard limit, then the residues aggregate to the product $m^{(2)}_{n-3}\cdot m^{(2)}_n$.  We lay the groundwork for investigations for $k=4$ and beyond; an ancillary notebook to the arXiv version of the paper is intended to facilitate this.  It is curious that the residues of $m^{(3)}_n$ and $m^{(3,NC)}_n$ seemingly coincide, given the base $\eta_{3,n-1,n}$.  Do they coincide for other bases?  Does the coincidence of the two extend to $k\ge 4$?

Our scattering equations calculation of residues uses parametrizations which degenerate to (torus orbits of) positroid subvarieties of $G(3,n)$.  A very natural question is to perform the same investigation for more general positroid subvarieties of $G(k,n)$; which among these would govern the behavior of residues of $m^{(k)}_n$, and conversely, which residues are calculated in this way?

In a different line of thought, it is known that the Global Schwinger formula for $m^{(k)}_n$, introduced in \cite{Cachazo:2020wgu}, expresses it as an integral transform of the positive tropical Grassmannian ${\rm Trop}^+G(k,n)$; but the integrand can also be interpreted in terms of a polytope, a certain generalized associahedron and in this interpretation, information about residues is stored in its faces.  In this way, we are finding a physical amplitude as a face of a polytope. This is reminiscent, for instance, of the way amplitudes appear as facets of the cosmological polytope \cite{Arkani-Hamed:2017fdk}. In that context, correlation functions become amplitudes once a residue on the ``energy'' pole is computed.

Other future directions include:
\begin{itemize}
	\item Exploring similar residue constructions in the context of the stringy integrals introduced by Arkani-Hamed, He and Lam \cite{Arkani-Hamed:2019mrd}.
	\item Connecting the various kinematic limits encountered in this work, especially the parallel hard limit, to the theory of likelihood degenerations of \cite{Agostini:2021rze}
	\item Exploring the generalized Veronese scalar amplitudes by finding the structure of their poles and generalized Feynman diagrams. Some first steps in the study of the corresponding scattering equations have been taken in \cite{Nick-forthcoming}.
\end{itemize}

\section*{Acknowledgements}

The authors thank Dani Kaufman and Bruno Umbert for useful discussions.  The second author thanks Nima Arkani-Hamed, Johannes Henn and Bernd Sturmfels for encouragement and support.   This research was supported in part by a grant from the Gluskin Sheff/Onex Freeman Dyson Chair in Theoretical Physics and by Perimeter Institute.  Research at Perimeter Institute is supported in part by the Government of Canada through the Department of Innovation, Science and Economic Development Canada and by the Province of Ontario through the Ministry of Colleges and Universities.  This research received funding from the European Research Council (ERC) under the European Union’s Horizon 2020 research and innovation programme (grant agreement No 725110), Novel structures in scattering amplitudes.

\bibliographystyle{jhep}
\bibliography{references}

\end{document}